\documentclass[12pt]{amsart}
\usepackage[left=15mm, right=15mm, top=15mm, bottom=15mm]{geometry}

\usepackage[english]{babel}
\usepackage[T1]{fontenc}
\usepackage[utf8]{inputenc}

\theoremstyle{plain}
    \newtheorem{theorem}{Theorem}[section]
    \newtheorem{lemma}[theorem]{Lemma}

\theoremstyle{definition}

\usepackage[
    draft = false,
    unicode = true,
    colorlinks = true,
    allcolors = blue,
    hyperfootnotes = true,
    citecolor = red
]{hyperref}
\usepackage{amsfonts,amsmath,amssymb,amscd,amsthm,url}
\usepackage[inline]{enumitem}
\usepackage{graphicx,epsf,afterpage, wrapfig}
\usepackage{soul}
\usepackage{multicol}
\usepackage{array}
\usepackage{braket}
\usepackage{epigraph}
\usepackage{rotating, floatflt}
\usepackage{xstring} % Пакет для надежной работы со строками

\usepackage[
backend=biber,
style=numeric-comp,
sorting=none,
giveninits=true
]{biblatex}
\usepackage{csquotes}

\usepackage{xcolor}

\usepackage{ulem}

\usepackage{tikz}
\usetikzlibrary{positioning, calc, intersections, through, arrows, matrix, chains, math}
\usetikzlibrary{arrows.meta}
\usetikzlibrary{shadows}
\usetikzlibrary{graphs}
\usetikzlibrary{graphs.standard}
\usetikzlibrary{patterns}
\usetikzlibrary{decorations.pathreplacing,calligraphy,backgrounds}
\DeclareFieldFormat{usera}{\href{https://arxiv.org/abs/#1}{\texttt{arXiv:#1}}}

\renewbibmacro*{finentry}{\printfield{usera}\newunit\finentry}

\renewbibmacro{in:}{%
  \iffieldundef{usera}
    {\printtext{\bibstring{in}\intitlepunct}}
    {}%
}

\newcommand\norm[1]{\ensuremath{\left\lVert#1\right\rVert}}
\newcommand\abs[1]{\ensuremath{\left\lvert#1\right\rvert}}

\DeclareMathOperator{\Ker}{Ker}

\DeclareMathOperator{\tr}{tr}

\newcommand{\Ocal}{\mathcal{O}}
\newcommand{\Pcal}{\mathcal{P}}
\newcommand{\Qcal}{\mathcal{Q}}

\newcommand{\Ucal}{\mathcal{U}}

\newcommand{\Abf}{{\ensuremath{\mathbf{A}}}}

\newcommand{\Ibf}{{\ensuremath{\mathbf{I}}}}

\newcommand{\Pbf}{{\ensuremath{\mathbf{P}}}}

\newcommand{\Rbf}{{\ensuremath{\mathbf{R}}}}

\newcommand{\Tbf}{{\ensuremath{\mathbf{T}}}}
\newcommand{\Ubf}{{\ensuremath{\mathbf{U}}}}
\newcommand{\Vbf}{{\ensuremath{\mathbf{V}}}}

\newcommand{\Phbf}{{\ensuremath{\mathbf{\Phi}}}}
\newcommand{\Xbf}{{\ensuremath{\mathbf{X}}}}
\newcommand{\Ybf}{{\ensuremath{\mathbf{Y}}}}
\newcommand{\Zbf}{{\ensuremath{\mathbf{Z}}}}

\newcommand{\defing}[1]{\textbf{\emph{\mathversion{bold}#1}}}

\newcommand{\R}{\ensuremath{\mathbb{R}}}

\newcommand{\Cx}{\ensuremath{\mathbb{C}}}

\renewcommand{\geq}{\geqslant}

\renewcommand{\leq}{\leqslant}

\newcounter{mcnt}

\newcounter{wordcnt}

\begin{document}

% \input{mains/meas-card}
% \input{mains/sing-chan}
% \input{mains/chan-von-neumann}
% \input{mains/ka-trans}
% \input{mains/ka-disc}
% \title{Quantum measurement \\ Kolmogorov--Arnold representation theorem}
\title{Algebraic Kolmogorov--Arnold representation theorem \\ for quantum measurement}

\author{Sviatoslav V. Dzhenzher}

\begin{abstract}
We establish an operational framework connecting the classical Kolmogorov--Arnold (KA) representation theorem to quantum information theory. By introducing and proving an algebraic, bounded-degree polynomial version of the theorem, we demonstrate that any target physical property of an unentangled multi-qubit product state can be exactly decomposed using a finite, fixed set of local ``inner'' observables and a shallow architecture of univariate polynomials. 

We further analyse the stability of this Quantum Kolmogorov--Arnold (QKA) representation under adversarial perturbations. In stark contrast to the pathological instabilities and severe reparameterisation sensitivities inherent to the classical Kolmogorov--Arnold representation theorem, our algebraic quantum framework exhibits remarkable resilience. We prove that the representation remains stable against bounded physical perturbations acting on the inner measurement operators, and show via the Heisenberg picture that it is inherently immune to adversarial quantum channel attacks acting on the input states. 
\end{abstract}

\thanks{\hspace{-4mm}
S.\,V. Dzhenzher: sdjenjer@yandex.ru. orcid: 0009-0008-3513-4312
\\
% V.\,Zh. Sakbaev: fumi2003@mail.ru. orcid: 0000-0001-8349-1738
% \\
% V.\,Zh. Sakbaev: Keldysh Institute of Applied Mathematics of Russian Academy of Sciences 125047, Miusskaya pl. 4, Moscow, Russia
% \\
% All authors:
Moscow Institute of Physics and Technology 141701, Institutskii per. 9, Dolgoprudny, Russia}

\maketitle
\thispagestyle{empty}


\vspace{3mm}
% 68T07  % Statistical learning theory (including deep learning, artificial neural networks)

\section{Introduction}\label{s:intro}

The challenge of efficiently approximating and representing multivariate functions lies at the heart of both classical and quantum machine learning. A foundational pillar of representation theory is the Kolmogorov--Arnold (KA) representation theorem \cite{Kolmogorov56, Arnold57, Kolmogorov57}, which states that any continuous multivariate function can be exactly decomposed into a finite composition of univariate functions and addition operations. For decades, this theorem remained primarily a theoretical curiosity due to the highly non-smooth and pathological nature of the underlying univariate functions. However, the optimistic views~\cite{Freedman24, Kurkova91} and the recent introduction of Kolmogorov--Arnold Networks (KANs)~\cite{KAN, KAN2} have revitalised this framework, demonstrating that parameterising univariate curves via splines or polynomials on the edges of a network provides an elegant, highly interpretable alternative to traditional Multi-Layer Perceptrons (MLPs).

As quantum computing expands into the domain of machine learning, translating successful classical architectures into the quantum regime has become a critical pursuit. This effort has led to the development of Quantum Kolmogorov--Arnold Networks (QKANs)~\cite{QKAN-enh-var, QKAN-china, QKAN-Ivashkov-etal, KANQAS}.
Existing QKAN architectures generally leverage quantum linear algebra primitives, such as the Quantum Singular Value Transformation (QSVT), to implement univariate edge activations directly on block-encoded matrices~\cite{QKAN-Ivashkov-etal}, or utilise Parameterised Quantum Circuits (PQCs) to parameterise activation functions within Born machine states~\cite{Werner-2025, QKAN-china}.

In this work, we approach the intersection of the Kolmogorov--Arnold theorem and quantum mechanics from a more fundamental operational perspective. Rather than building a specific machine learning ansatz, we aim to formulate a Quantum Kolmogorov--Arnold (QKA) representation theorem directly for quantum observables and states (see the main results: theorems~\ref{t:qka-fact}, \ref{t:qka-stab-meas} and~\ref{t:qka-stab-q-ch}).
A closely related foundational bridge has been recently established by Freedman \cite{Freedman25}, who connected the Kolmogorov--Arnold representation to the study of recurrent quantum dynamics and the detection of hidden tensor factorisations.
Note that, while classical measurement theory is traditionally viewed as commutative, alternative operational frameworks, such as modern Koopman mechanics \cite{Brunton-2022}, can introduce non-commutativity into classical structures via the Poisson bracket \cite{Morgan-2020}. In contrast to these approaches that establish formal isomorphisms by extending classical mechanics, our algebraic framework directly addresses quantum measurements via a polynomial version of the KA theorem.
Following this operational direction, in Section~\ref{s:main}, we use an algebraic, polynomial version of the KA theorem (Lemma~\ref{l:ka-poly}), and utilise it to construct universal, fixed ``inner'' observables that can represent any target property of a multi-qubit product state using a shallow combination of univariate polynomials, as in Theorem~\ref{t:qka-fact}.

After that, in Section~\ref{s:adversary}, following \cite{DzhenzherFreedman-25}, we investigate our QKA representation theorem in terms of stability under adversarial perturbations. We investigate two modes: when an adversary is a quantum channel acting on the state (Theorem~\ref{t:qka-stab-q-ch}), and when an adversary is a small perturbation of our ``inner'' observables (Theorem~\ref{t:qka-stab-meas}).
We find that, as opposed to \cite{DzhenzherFreedman-25}, we are able to withstand practically any adversarial attack (in the case of perturbations of inner observables, it should be a small enough perturbation).

Critically, while our construction establishes a valid representation theorem for static state verification and property estimation, it also exposes fundamental structural limitations. As we discuss in the concluding Section~\ref{s:conclusion}, attempts to generalise this exact framework to describe full unitary evolutions (or, generally speaking, quantum channels) face severe barriers: in those dynamic regimes, straightforward analogues of the Kolmogorov--Arnold theorem either collapse into trivial mathematical identities or become inherently incorrect due to the non-commutative and non-linear constraints of quantum dynamics.

\section{Background and main results}\label{s:main}

To establish the framework for our main results, we first review the mathematical statements of the classical Kolmogorov--Arnold theorem,
% and its polynomial restrictions,
before defining the notation for the quantum systems under consideration.

The classical Kolmogorov--Arnold representation theorem guarantees that for any $n>1$, there exist $n(2n+1)$ universal continuous univariate functions $\phi_{j,k}\colon[\,0,1\,]\to[\,0,1\,]$ such that for any continuous function $f\colon [\,0,1\,]^n \to \R$, there exist $2n+1$ continuous univariate functions $g_j\colon\R\to\R$ such that for any \(x_1,\ldots,x_n\in[\,0,1\,]\),
\[
    f(x_1, \dots, x_n) = \sum_{j=1}^{2n+1} g_j \left( \sum_{k=1}^n \phi_{j,k}(x_k) \right).
\]
In modern KAN architectures, the functions $\phi_{j,k}$ are learned or approximated using smooth functions or Chebyshev polynomials. In our algebraic treatment, we constrain this representation to polynomials of a bounded degree. Specifically, if a multivariate function is a polynomial of degree at most $d$, it can be exactly decomposed using universal \emph{linear} combinations inside the inner layer, which forms the basis of Lemma~\ref{l:ka-poly}.

Let $\Cx^{2^n}$ denote the Hilbert space of $n$ qubits. We denote the space of \defing{observables} (bounded Hermitian operators acting on this space) as $\Ocal(\Cx^{2^n})$.
A general $n$-qubit \defing{density matrix} $\rho$ is a positive semi-definite Hermitian matrix with a unit trace (the latter is required in physics; we will need only the fact that all states lie in some ball).
For the scope of this representation theorem, we consider an unentangled, factorised input state $\rho_{\text{in}}$ given by the tensor product:
\[
    \rho_{\text{in}} = \rho_1 \otimes \rho_2 \otimes \dots \otimes \rho_n,
\]
where each $\rho_i$ represents the local state of the $i$-th qubit. 
As we will see, the expected value of any target observable $\Tbf \in \Ocal(\Cx^{2^n})$ with respect to $\rho_{\text{in}}$ can be viewed as a multivariate polynomial function mapping the $3n$ real coordinates to $\R$. Our goal is to decompose this expected value using a universal set of local quantum measurements.

\begin{theorem}[Quantum KA]\label{t:qka-fact}
    Let $n\geq 1$ be the number of qubits, and \(m \geq \binom{4n-1}{n}\).
    Then there exist universal ``inner'' observables
    \[
        \Abf_{j,k} \in \Ocal(\Cx^2),\qquad
        j=1,\ldots,m, \quad k=1,\ldots,n,
    \]
    such that for any ``target'' observable \(\Tbf\in \Ocal(\Cx^{2^n})\), there exist ``outer'' polynomials \(g_1,\ldots,g_m\colon\R\to\R\) of degree at most $n$, such that for any factorised state \(\rho_{\text{in}}=\rho_1\otimes\ldots\otimes\rho_n\),
    \[
        \tr(\Tbf\rho_{\text{in}}) = \sum_{j=1}^m g_j\bigl(\tr(\Abf_{j,1}\rho_1)+\ldots+\tr(\Abf_{j,n}\rho_n)\bigr).
    \]
\end{theorem}

We will prove Theorem~\ref{t:qka-fact},
% and~\ref{t:qka-mixed},
using the following algebraic version of KA.

\begin{lemma}[Polynomial KA {\cite[Theorem~1]{Schinzel-2002}}]\label{l:ka-poly}
    Let $n,d\geq 1$ and \(m \geq \binom{n+d-1}{d}\).
    Then there exist universal linear functions \(\ell_1,\ldots,\ell_m\colon\R^n\to\R\) of the form
    \[
        \ell_j(x_1,\ldots,x_n) = \sum_{k=1}^n \ell_{j,k}x_k,
        \qquad j=1,\ldots,m,
    \]
    such that for any polynomial $p\colon\R^n\to\R$ of degree at most $d$, there exist polynomials \(q_1,\ldots,q_m\colon\R\to\R\) of degree at most $d$ such that for any \(x_1,\ldots,x_n\in\R\),
    \[
        p(x_1,\ldots,x_n) = \sum_{j=1}^m q_j\left(\ell_j(x_1,\ldots,x_n)\right).
    \]
    Moreover, one may take
    \[
        \ell_{j,k} = \pi_k^j,
    \]
    where \(\pi_k\) is the $k$-th prime number.
\end{lemma}

In the classical KA theorem, the ``inner'' functions $\phi_{j,k}$ cannot be chosen linear, and they are inherently non-smooth, non-differentiable fractals~\cite{Lorentz-Golitschek-Makovoz}.
A fundamental result by Vitushkin~\cite{Vitushkin-1954, Vitushkin-overview} established that if one restricts the representation to continuously differentiable functions, the exact superposition formula fails for generic multivariate functions due to strict information-theoretic and smooth dimension barriers. 

Lemma~\ref{l:ka-poly} states that if the ``target'' function $f$ is polynomial (one may say, is good enough), the inner functions can indeed be chosen strictly linear, which provides an algebraic shortcut that avoids the pathological fractal geometry of the classical representation.
Note that the analogous theorem could be proven via the polarisation identity, which for two variables looks like \(x_1x_2=\frac{(x_1+x_2)^2 - (x_1-x_2)^2}{4}\), but it would demand greater~$m$.
Although the lemma is known \cite{Schinzel-2002}, we provide its complete proof since it is short and slightly simpler than in \cite{Schinzel-2002}.

\begin{proof}[Proof of Lemma~\ref{l:ka-poly}]
    Take \(\ell_{j,k} = \pi_k^j\). Let
    \[
        q_j(t) = \sum_{r=0}^d q_{j,r}t^r.
    \]
    Let \(p_{r_1,\ldots,r_n}\) be the coefficient of the monomial \(x_1^{r_1}\ldots,x_n^{r_n}\) in the polynomial $p$.
    The contribution of $q_j$ to $p_{r_1,\ldots,r_n}$ is determined only by the term \(q_{j,r}t^r\) for \(r:=r_1+\ldots+r_n\), and is equal to
    \[
        \binom{r}{r_1,\ldots,r_n}q_{j,r}\prod_{k=1}^n\pi_k^{jr_k}.
    \]
    Consider the Vandermonde matrix \(V\) with entries
    \[
        V_{(r_1,\ldots,r_n),j} := \left(\prod_{k=1}^n\pi_k^{r_k}\right)^j.
    \]
    Thus, we obtain the system of \(\binom{n+d}{n}\) equations of the form
    \[
        \frac{p_{r_1,\ldots,r_n}}{\binom{r}{r_1,\ldots,r_n}} =
        \sum_{j=1}^m q_{j,r}V_{(r_1,\ldots,r_n),j}.
    \]
    It remains to notice that, actually, we have $d+1$ independent systems for different $r$.
    For a fixed $r$, the Vandermonde matrix $V$ has full rank (by the prime factorisation theorem), and thus the system has a solution, since the number of variables is
    \[
        m \geq \binom{n+d-1}{d} = \binom{n+d-1}{n-1} \geq \binom{n+r-1}{n-1}.
    \]
\end{proof}

\begin{proof}[Proof of Theorem~\ref{t:qka-fact}]
    Denote by \(\Xbf,\Ybf,\Zbf\) the Pauli matrices.
    It is well known that for any density matrix \(\rho_i\) of one qubit, there is a representation in the form
    \[
        \rho_i = \frac{1}{2}\left(\Ibf + x_i\Xbf + y_i\Ybf + z_i\Zbf\right),
    \]
    where \((x_i,y_i,z_i)\in B^3 = \{(x,y,z)\in\R^3:x^2+y^2+z^2\leq 1\}\).
    Thus, we may define the continuous function \(f\colon (B^3)^n\to \R\) on the compact metric space \((B^3)^n\) by
    \[
        f(x_1,y_1,z_1,\ldots,x_n,y_n,z_n) := \tr(\Tbf\rho_1\otimes\ldots\otimes\rho_n).
    \]

    In this paragraph, we show that $f$ is the polynomial of degree at most $n$.
    Indeed, the target observable can be represented as
    \[
        \Tbf = \sum_{\alpha_1,\ldots,\alpha_n=0}^3 \tau_{\alpha_1,\ldots,\alpha_n} \sigma_{\alpha_1}\otimes\ldots\otimes\sigma_{\alpha_n},
    \]
    where \(\sigma_i\in\{\Ibf,\Xbf,\Ybf,\Zbf\}\) and \(\tau_{\alpha_1,\ldots,\alpha_n}\in\R\).
    Then,
    \begin{multline*}
        f(x_1,y_1,z_1,\ldots,x_n,y_n,z_n) = \tr(\Tbf\rho_1\otimes\ldots\otimes\rho_n) = \\ =
        \sum_{\alpha_1,\ldots,\alpha_n=0}^3 \tau_{\alpha_1,\ldots,\alpha_n} \tr(\sigma_{\alpha_1}\otimes\ldots\otimes\sigma_{\alpha_n}\cdot\rho_1\otimes\ldots\otimes\rho_n) = \\ =
        \sum_{\alpha_1,\ldots,\alpha_n=0}^3 \tau_{\alpha_1,\ldots,\alpha_n} \prod_{k=1}^n\tr(\sigma_{\alpha_k}\rho_k),
    \end{multline*}
    where \(\tr(\sigma_{\alpha_k}\rho_k) \in \{1,x_k,y_k,z_k\}\).
    Thus, $f$ is the polynomial of $3n$ variables and degree at most $n$.

    Now, for any \(j=1,\ldots,m\) and \(k=1,\ldots,n\), let an observable \(\Abf_{j,k}\) be of the form
    \[
        \Abf_{j,k} = a_{j,k,x}\Xbf + a_{j,k,y}\Ybf + a_{j,k,z}\Zbf.
    \]
    Then, it is well known that
    \[
        \tr(\Abf_{j,k}\rho_k) = a_{j,k,x}x_k + a_{j,k,y}y_k + a_{j,k,z}z_k.
    \]
    Thus, we obtain that each $g_j$ has as its argument the linear combination of variables $x_k,y_k,z_k$, and this theorem is a particular case of Lemma~\ref{l:ka-poly} for $3n$ variables and $d=n$.
\end{proof}

\section{Quantum Kolmogorov--Arnold stability}\label{s:adversary}

Now we turn to QKA stability, which was first described in \cite{DzhenzherFreedman-25} for the classical KA.
It was shown there that an outer target function $f$ can withstand adversarial attacks (a.k.a. reparameterisations of the inner layers), if these attacks are homeomorphisms \(\R^{2n+1}\to\R^{2n+1}\) and their at most countable collection is reported in advance.
Here, we may distinguish two global different cases: if some adversarial attack acts on our inner measurement operators \(\Abf_j\), or if some adversarial quantum channel acts on our input state \(\rho_{\text{in}}\).

In stark contrast, our algebraic framework rules out such pathological behaviour. Because the inner operations in Theorem~\ref{t:qka-fact} are formulated via linear quantum measurements on finite-dimensional space, we can establish robust stability bounds.
Let us formalise the ``adversary'' not as a topological homeomorphism, but as a physical perturbation acting on the universal inner observables.

\begin{theorem}[QKA Stability under Measurement Attacks]\label{t:qka-stab-meas}
    Let $n\geq 1$ be the number of qubits, and \(m \geq \binom{4n-1}{n}\).
    Then there exist universal ``inner'' observables
    \[
        \Abf_{j,k} \in \Ocal(\Cx^2),\qquad
        j=1,\ldots,m, \quad k=1,\ldots,n,
    \]
    such that for any ``target'' observable \(\Tbf\in \Ocal(\Cx^{2^n})\), any sufficiently small \(\varepsilon>0\),
    and any disturbances
    \[
        \Delta_{j,k} \in \Ocal(\Cx^2),
        \qquad        \norm{\Delta_{j,k}} \leq \varepsilon,
    \]
    there exist ``outer'' polynomials \(g_1,\ldots,g_m\colon\R\to\R\) of degree at most $n$ such that for any factorised state \(\rho_{\text{in}}=\rho_1\otimes\ldots\otimes\rho_n\),
    \[
        \tr(\Tbf\rho_{\text{in}}) = \sum_{j=1}^m g_j\bigl(\tr((\Abf_{j,1}+\Delta_{j,1})\rho_1)+\ldots+\tr((\Abf_{j,n}+\Delta_{j,n})\rho_n)\bigr).
    \]
\end{theorem}

Analogously to the previous case, this follows from the polynomial version below.

\begin{lemma}
    Let $n,d\geq 1$ and \(m \geq \binom{n+d-1}{d}\).
    Then there exist universal linear functions \(\ell_1,\ldots,\ell_m\colon\R^n\to\R\) of the form
    \[
        \ell_j(x_1,\ldots,x_n) = \sum_{k=1}^n \ell_{j,k}x_k,
        \qquad j=1,\ldots,m,
    \]
    such that for any polynomial $p\colon\R^n\to\R$ of degree at most $d$, and any sufficiently small \(\delta\ell_{j,k}\in\R\), there exist polynomials \(q_1,\ldots,q_m\colon\R\to\R\) of degree at most $d$ such that for any \(x_1,\ldots,x_n\in\R\),
    \[
        p(x_1,\ldots,x_n) = \sum_{j=1}^m q_j\left(\sum_{k=1}^n (\ell_{j,k}+\delta\ell_{j,k})x_k\right).
    \]
    % Moreover, one may take
    % \[
    %     \ell_{j,k} = \pi_k^j,
    % \]
    % where \(\pi_k\) is the $k$-th prime number.
\end{lemma}

\begin{proof}
    Take \(\ell_{j,k}\) from Lemma~\ref{l:ka-poly}, and denote
    \[
        \ell_j := \ell_j(x_1,\ldots,x_n) = \sum_{k=1}^n \ell_{j,k}x_k,
        \quad j=1,\ldots,m.
    \]
    Let \(\Pcal\) be the Banach space of polynomials of degree at most $d$ of $n$ variables, and \(\Qcal\) be the Banach space of tuples \(\vec q=(q_1,\ldots,q_m)\) of polynomials of degree at most $d$, both with the sup-norm on some compact cube \([\,-K,K\,]^n\) and compact interval \([\,-K,K\,]\), respectively.
    By equivalence of norms in finite-dimensional spaces, all resulting operator bounds are independent of this choice up to a multiplicative constant.
    
    Define the linear bounded operator \(\Pbf\colon \Qcal\to\Pcal\) by the rule
    \[
        \Pbf(\vec q)(x_1,\ldots,x_n) := \sum_{j=1}^m q_j(\ell_j).
    \]
    By Lemma~\ref{l:ka-poly}, this operator is surjective.
    By the bounded inverse theorem, there exists an inverse linear bounded operator
    \(\Rbf = \left(\Pbf|_{(\Ker\Pbf)^\perp}\right)^{-1}\colon \Pcal\to(\Ker\Pbf)^\perp\subset \Qcal\),
    which is simply a right inverse to \(\Pbf\).
    (Here by \((\Ker\Pbf)^\perp\) we mean a linear complement given by \(\Ker\Pbf \oplus (\Ker\Pbf)^\perp=\Qcal\).)
    
    Now, fix any polynomial \(p\in\Pcal\) and disturbances \(\delta\ell_{j,k} \in O_\varepsilon(0)\) with $\varepsilon>0$ to be specified later.
    Denote
    \[
        \delta\ell_j = \delta\ell_j(x_1,\ldots,x_n) := \sum_{k=1}^n\delta\ell_{j,k}x_k.
    \]
    Define the linear bounded operator
    \(\delta\Pbf\colon \Qcal\to\Pcal\) by the rule
    \[
        \delta\Pbf(\vec q)(x_1,\ldots,x_n) := \sum_{j=1}^m \bigl(q_j(\ell_j+\delta\ell_j) - q_j(\ell_j)\bigr).
    \]
    
    At last, we are ready to define \(\vec{q}=(q_1,\ldots,q_m)\in\Qcal\).
    Apply Lemma~\ref{l:ka-poly} without any disturbances in order to obtain \(q_{j,0}\) such that
    \[
        p(x_1,\ldots,x_n) = \sum_{j=1}^m q_{j,0}(\ell_j).
    \]
    In our new notations, this is just
    \[
        \vec{q}_0 = (q_{1,0},\ldots,q_{m,0}) = \Rbf(p).
    \]
    Here goes a trick: below, we will show that $\varepsilon$ can be chosen sufficiently small, so that the operator \(\Rbf\circ\delta\Pbf\) is a contraction.
    Then, by the Banach fixed-point theorem, there exists a fixed-point \(\vec{q} \in \Qcal\) such that
    \[
        \vec q = \vec{q}_0-\Rbf\left(\delta\Pbf(\vec{q})\right).
    \]
    Acting by \(\Pbf\), we obtain
    \[
        \Pbf(\vec q) = p - \delta\Pbf(\vec{q}),
    \]
    which means that
    \[
        p(x_1,\ldots,x_n) = \sum_{j=1}^m q_j(\ell_j) + \sum_{j=1}^m \bigl(q_j(\ell_j+\delta\ell_j) - q_j(\ell_j)\bigr) =
        \sum_{j=1}^m q_j(\ell_j+\delta\ell_j).
    \]
    This means that what we have found is exactly the desired tuple $\vec{q}$.

    It remains to estimate \(\varepsilon\).
    It is clear that
    \[
        \abs{\delta\ell_j} \leq nK\varepsilon
        \quad\text{for any}\quad j=1,\ldots,m.
    \]
    Since
    \[
        \abs{q_j(\ell_j+\delta\ell_j) - q_j(\ell_j)} \leq
        \sum_{r=1}^d \frac{1}{r!}\abs{q_j^{(r)}(\ell_j)(\delta\ell_j)^r},
    \]
    and the differential operator is linear and bounded on $\Qcal$, we obtain that, for some positive constant $C=C(n,d,m,K,\ell_1,\ldots,\ell_m)$,
    \[
        \norm{\delta\Pbf(\vec q)} \leq C\varepsilon \norm{\vec q}.
    \]
    Thus, it remains to take any \(\varepsilon < \min(1,\frac{1}{C\norm{\Rbf}})\).
\end{proof}

\begin{proof}[Proof of Theorem~\ref{t:qka-stab-meas}]
    The proof is analogous to the deduction of Theorem~\ref{t:qka-fact} from Lemma~\ref{l:ka-poly}.
    The only difference is that \(\Delta_{j,k}\) may contain the shift \(c\Ibf\), which gives some constant shift \(c_j\) to $g_j$.
    For the proof, we need firstly to solve the case for \(c_j=0\) and obtain some $g_j'$, and after that put
    \[
        g_j(t) := g_j'(t-c_j).
    \]
\end{proof}

Now we focus to stability under quantum attacks. Recall that a quantum channel is a linear continuous operator that maps quantum states in another quantum states.

\begin{theorem}[QKA Stability under Quantum Attacks]\label{t:qka-stab-q-ch}
    Let $n\geq 1$ be the number of qubits, and \(m \geq \binom{4n-1}{n}\).
    Then there exist universal ``inner'' observables
    \[
        \Abf_{j,k} \in \Ocal(\Cx^2),\qquad
        j=1,\ldots,m, \quad k=1,\ldots,n,
    \]
    such that for any ``target'' observable \(\Tbf\in \Ocal(\Cx^{2^n})\) and any quantum channel \(\Phbf\),
    there exist ``outer'' polynomials \(g_1,\ldots,g_m\colon\R\to\R\) of degree at most $n$ such that for any factorised state \(\rho_{\text{in}}=\rho_1\otimes\ldots\otimes\rho_n\),
    \[
        \tr(\Tbf\Phbf(\rho_{\text{in}})) = \sum_{j=1}^m g_j\bigl(\tr(\Abf_{j,1}\rho_1)+\ldots+\tr(\Abf_{j,n}\rho_n)\bigr).
    \]
\end{theorem}

\begin{proof}
    It is well known that instead of the Schr\"{o}dinger representation $\Phbf$, we may consider the Heisenberg representation \(\Phbf^*\), acting on the space of observables by the rule
    \[
        \tr(\Tbf\Phbf(\rho_{\text{in}})) = \tr(\Phbf^*(\Tbf)\rho_{\text{in}}).
    \]
    This makes the theorem a trivial corollary of Theorem~\ref{t:qka-fact}, applied to the observable \(\Phbf^*(\Tbf)\).
\end{proof}

\section{Conclusions and future outlook}\label{s:conclusion}

In this work, we have established an operational bridge between the classical Kolmogorov--Arnold representation theorem and quantum mechanics. By formulating and proving an algebraic, polynomial version of the theorem (Lemma~\ref{l:ka-poly}), we demonstrated that any target property of a multi-qubit unentangled state can be exactly decomposed using a finite, fixed set of local ``inner'' observables combined with shallow univariate polynomials (Theorem~\ref{t:qka-fact}). 

A major finding of our approach lies in its robustness against noise. While classical Kolmogorov--Arnold networks are known to exhibit pathological instability and severe sensitivity under continuous reparameterisations of the hidden layers~--- as rigorously studied in \cite{DzhenzherFreedman-25}~--- our algebraic quantum framework circumvents these vulnerabilities.
We proved that the proposed representation can withstand physical perturbations up to a bounded threshold $\varepsilon$ acting on the inner measurement operators (Theorem~\ref{t:qka-stab-meas}). Furthermore, by shifting to the Heisenberg picture, we showed that our setup is inherently immune to adversarial quantum channel attacks acting on the input states (Theorem~\ref{t:qka-stab-q-ch}).

Critically, our positive results heavily rely on the assumption that the input state $\rho_{\text{in}}$ remains factorised. If we extend this framework to general, non-factorised (entangled) states, the polynomial structure changes drastically. Instead of mapping $3n$ independent local coordinates, the expectation value $\tr(\Tbf\rho_{\text{in}})$ begins to depend on the full set of exponentially many density matrix coefficients. This leads to a severe explosion in the number of unknown polynomial coefficients and degrees, which effectively prevents us from obtaining any structurally meaningful or tight representation bound in the presence of global entanglement. Probably, for this case, we need a technique that does not use a polynomial structure and is closer to the classical KA proof.

Moreover, while this framework provides a powerful tool for static state verification and property estimation, it also exposes key structural boundaries.
While physics-informed KANs have been successfully deployed to forecast specific temporal quantum dynamics via Ehrenfest theorems \cite{QKAN-time-series} or parameterise strongly correlated wave functions \cite{QKAN-spins}, our work exposes the boundary where a native, exact operator algebraic analogue for full unitary transformations remains absent. Specifically, attempts to generalise this exact representation to describe full quantum dynamics face severe mathematical barriers. When moving from static state properties to full unitary evolutions $\Ubf \in \mathcal{U}(\Cx^{2^n})$ or general quantum channels $\Phbf$, straightforward analogues of the Kolmogorov--Arnold theorem either collapse into trivial identities or break down entirely. This failure stems directly from the non-commutative structure of the underlying operator algebra and the non-linear constraints governing quantum evolution.
Interestingly, alternative approaches using the Kolmogorov--Arnold theorem have been proposed \cite{Freedman25} to bypass these exact constraints, focusing instead on algorithms for detecting low-volume approximate recurrence in unitary dynamics.

In order to illustrate these dynamic limitations, we present below two straightforward quantum versions of the Kolmogorov--Arnold theorem for unitary evolutions, which turn out to be trivial and incorrect.

\begin{itemize}
    \item \textbf{The trivial version:} \emph{For the number $n \geq 1$ of main qubits, there exist a number $m = m(n)$ of ancilla qubits and a universal ancilla state $\psi \in \Cx^{2^m}$, such that for any target operator \(\Ubf\in \Ucal(\Cx^{2^n})\), there exists an operator \(\Vbf\in \Ucal(\Cx^{2^{n+m}})\) such that for any state $\rho$ on $n$ qubits,}
    \[
        \Ubf\rho\Ubf^\dagger = \tr_{ancilla} \left(\Vbf \cdot \left(\rho \otimes \ket{\psi}\bra{\psi}\right) \cdot \Vbf^\dagger\right).
    \]
    This statement is contextually trivial. Since the state $\psi$ is fixed and universal, one can simply choose $\Vbf = \Ubf \otimes \Ibf$, where $\Ibf$ is the identity operator on the ancilla space. Upon taking the partial trace, the state $\ket{\psi}$ trivially traces out to unity, yielding $\Ubf\rho\Ubf^\dagger$ identically. This construction completely bypasses any meaningful structural decomposition or reduction.
    
    \item \textbf{The incorrect version:} Alternatively, one might attempt to swap the quantifiers to make the inner operator universal. Specifically, \emph{for the number $n \geq 1$ of main qubits, there exist a number $m = m(n)$ of ancilla qubits and a universal operator $\Vbf \in \Ucal(\Cx^{2^{n+m}})$, such that for any target operator $\Ubf \in \Ucal(\Cx^{2^n})$, there exists a state $\psi \in \Cx^{2^m}$ such that for any state $\rho$ on $n$ qubits, the same identity holds:}
    \[
        \Ubf\rho\Ubf^\dagger = \tr_{ancilla} \left(\Vbf \cdot \left(\rho \otimes \ket{\psi}\bra{\psi}\right) \cdot \Vbf^\dagger\right).
    \]
    While this formulation might initially seem correct, it is wrong by the No-Programming Theorem \cite{Nielsen-Chuang-noprog}.
    This theorem states that no such deterministic universal quantum processor $\Vbf$ exists, for which any unitary evolution $\Ubf$ is implementable by some quantum program $\ket{\psi}$.
    Distinct non-equivalent unitaries require orthogonal program states, whereas a finite-dimensional program register cannot contain mutually orthogonal states for the full continuous family \(\Ucal(\Cx^{2^n})\).
\end{itemize}

Consequently, a fascinating and critical direction for future research is the formulation of a true, native quantum analogue of the Kolmogorov--Arnold theorem for unitary evolutions that natively respects the Lie algebraic structure of quantum gates. Developing a framework that can decompose complex, highly-entangling multi-qubit unitaries into a composition of univariate dynamic primitives remains an open challenge. Resolving this problem could pave the way for entirely new classes of highly interpretable, hardware-efficient quantum machine learning architectures.

\section*{Data availability}

The research has no associated data.

\section*{Conflict of interest}

The author confirms that he has no conflict of interest.

\printbibliography
% % \printbibitembibliography

\end{document}